\documentclass[conference,letterpaper,romanappendices,onecolumn,draft]{ieeeconf}
\let\proof\relax   

\usepackage{amsthm,xpatch}
\usepackage{amsmath,amsfonts}
\usepackage{cite}
\usepackage{amssymb}
\usepackage{bm}
\usepackage{dsfont}
\usepackage{graphicx, subfigure}
\usepackage{color}
\usepackage{breqn}
\usepackage{mathtools}
\usepackage{bbm}
\usepackage{latexsym}
\usepackage[ruled, linesnumbered]{algorithm2e}
\usepackage{accents}
\usepackage{tikz}
\usepackage{multirow}
\usepackage{verbatim}

\newtheorem{lemma}{Lemma}
\newtheorem{theorem}{Theorem}

\newtheorem{example}{Example}

\newcommand*{\transpose}{%
  {\mathpalette\@transpose{}}%
}

\IEEEoverridecommandlockouts

\begin{document}

\newcommand{\SB}[3]{
\sum_{#2 \in #1}\biggl|\overline{X}_{#2}\biggr| #3
\biggl|\bigcap_{#2 \notin #1}\overline{X}_{#2}\biggr|
}

\newcommand{\Mod}[1]{\ (\textup{mod}\ #1)}

\newcommand{\overbar}[1]{\mkern 0mu\overline{\mkern-0mu#1\mkern-8.5mu}\mkern 6mu}

\makeatletter
\newcommand*\nss[3]{%
  \begingroup
  \setbox0\hbox{$\m@th\scriptstyle\cramped{#2}$}%
  \setbox2\hbox{$\m@th\scriptstyle#3$}%
  \dimen@=\fontdimen8\textfont3
  \multiply\dimen@ by 4             
  \advance \dimen@ by \ht0
  \advance \dimen@ by -\fontdimen17\textfont2
  \@tempdima=\fontdimen5\textfont2  
  \multiply\@tempdima by 4
  \divide  \@tempdima by 5          
  \ifdim\dimen@<\@tempdima
    \ht0=0pt                        
    \@tempdima=\fontdimen5\textfont2
    \divide\@tempdima by 4          
    \advance \dimen@ by -\@tempdima 
    \ifdim\dimen@>0pt
      \@tempdima=\dp2
      \advance\@tempdima by \dimen@
      \dp2=\@tempdima
    \fi
  \fi
  #1_{\box0}^{\box2}%
  \endgroup
  }
\makeatother

\makeatletter
\renewenvironment{proof}[1][\proofname]{\par
  \pushQED{\qed}%
  \normalfont \topsep6\p@\@plus6\p@\relax
  \trivlist
  \item[\hskip\labelsep
        \itshape
    #1\@addpunct{:}]\ignorespaces
}{%
  \popQED\endtrivlist\@endpefalse
}
\makeatother

\makeatletter
\newsavebox\myboxA
\newsavebox\myboxB
\newlength\mylenA

\newcommand*\xoverline[2][0.75]{%
    \sbox{\myboxA}{$\m@th#2$}%
    \setbox\myboxB\null
    \ht\myboxB=\ht\myboxA%
    \dp\myboxB=\dp\myboxA%
    \wd\myboxB=#1\wd\myboxA
    \sbox\myboxB{$\m@th\overline{\copy\myboxB}$}
    \setlength\mylenA{\the\wd\myboxA}
    \addtolength\mylenA{-\the\wd\myboxB}%
    \ifdim\wd\myboxB<\wd\myboxA%
       \rlap{\hskip 0.5\mylenA\usebox\myboxB}{\usebox\myboxA}%
    \else
        \hskip -0.5\mylenA\rlap{\usebox\myboxA}{\hskip 0.5\mylenA\usebox\myboxB}%
    \fi}
\makeatother

\xpatchcmd{\proof}{\hskip\labelsep}{\hskip3.75\labelsep}{}{}

\pagestyle{plain}

\title{\fontsize{20}{28}\selectfont Multi-Server Private Information Retrieval\\ with Coded Side Information}

\author{Fatemeh Kazemi, Esmaeil Karimi, Anoosheh Heidarzadeh, and Alex Sprintson\thanks{The authors are with the Department of Electrical and Computer Engineering, Texas A\&M University, College Station, TX 77843 USA (E-mail: \{fatemeh.kazemi, esmaeil.karimi, anoosheh, spalex\}@tamu.edu).}\thanks{A short version of this work was presented at the 16th Canadian Workshop on Information Theory (CWIT'19), Hamilton, Ontario, Canada, June 2019.}\thanks{This material is based upon work supported by the National Science Foundation under Grants No.~1718658 and 1642983.}}

\maketitle 

\thispagestyle{plain}

\begin{abstract}
In this paper, we study the multi-server setting of the \emph{Private Information Retrieval with Coded Side Information (PIR-CSI)} problem. In this problem, there are $K$ messages replicated across $N$ servers, and there is a user who wishes to download one message from the servers without revealing any information to any server about the identity of the requested message. The user has a side information which is a linear combination of a subset of $M$ messages in the database. The parameter $M$ is known to all servers in advance, whereas the indices and the coefficients of the messages in the user's side information are unknown to any server \emph{a priori}. 

We focus on a class of PIR-CSI schemes, referred to as \emph{server-symmetric schemes}, in which the queries/answers to/from different servers are symmetric in structure. We define the \emph{rate} of a PIR-CSI scheme as its minimum download rate among all problem instances, and define the \emph{server-symmetric capacity} of the PIR-CSI problem as the supremum of rates over all server-symmetric PIR-CSI schemes. Our main results are as follows: (i) when the side information is not a function of the user's requested message, the capacity is given by ${(1+{1}/{N}+\dots+{1}/{N^{\left\lceil \frac{K}{M+1}\right\rceil -1}})^{-1}}$ for any ${1\leq M\leq K-1}$; and (ii) when the side information is a function of the user's requested message, the capacity is equal to $1$ for $M=2$ and $M=K$, and it is equal to ${N}/{(N+1)}$ for any ${3 \leq M \leq K-1}$. The converse proofs rely on new information-theoretic arguments, and the achievability schemes are inspired by our recently proposed scheme for single-server PIR-CSI as well as the Sun-Jafar scheme for multi-server PIR.
\end{abstract}



\section{introduction}
In the Private Information Retrieval (PIR) problem, there is a user who wishes to download a single or multiple messages belonging to a database with copies stored on a single or multiple servers, while protecting the identity of the demanded message(s) from the server(s)~\cite{sun2017capacity,banawan2018capacity}. This setup was recently extended to the settings wherein the user has some side information (unknown to the server(s)) about the messages in the database~\cite{kadhe2017private,heidarzadeh,li2018single, heidarzadeh2018capacity,heidarzadeh2019capacity,tandon2017capacity, wei2018fundamental,li2018converse,chen2017capacity,shariatpanahi2018multi,kazemi2019single}.

For the PIR problem in the presence of side information, the savings in the download cost (i.e., the amount of information downloaded
from the server(s)) depend on whether the user wants to achieve $W$\textit{-privacy} (i.e., only the privacy of the requested message(s) needs to be protected), or $(W,S)$\textit{-privacy} (i.e., the privacy of both the requested message(s) and the messages in the side information need to be protected). 

The settings in which the side information is a subset of messages, referred to as \emph{PIR with Side Information (PIR-SI)} when $W$-privacy is required, and \emph{PIR with Private Side Information (PIR-PSI)} when $(W,S)$-privacy is required, were studied in~\cite{kadhe2017private,heidarzadeh,li2018single,li2018converse,chen2017capacity,shariatpanahi2018multi}. Recently, in~\cite{heidarzadeh2018capacity} and \cite{heidarzadeh2019capacity}, we studied the single-server case of a related problem wherein the side information is a linear combination of a subset of messages. This problem is referred to as \textit{PIR with Coded Side Information (PIR-CSI)} when $W$-privacy is required, and \textit{PIR with Private Coded Side Information (PIR-PCSI)} when $(W,S)$-privacy is required. 

In this work, we consider the multi-server case of the PIR-CSI problem. In this problem, there is a database of $K$ messages replicated across $N$ servers, and there is a user who wants to download a single message from the servers, while revealing no information about the requested message to any server. We assume that the user has a linear combination of a subset of $M$ messages in the database as side information. Also, we assume that the side information size $M$ is known to all servers in advance, but no server knows the indices and the coefficients of the messages in the user's side information in advance. These assumptions are motivated by several practical scenarios. For instance, the user could have obtained their side information from a trusted entity with limited knowledge about the database, 
or from the information locally stored in their cache. 
\subsection{Main Contributions}


We focus on a class of PIR-CSI schemes, which we refer to as \emph{server-symmetric}, where the queries/answers to/from different servers are symmetric in structure. Such schemes are of particular interest in practice for their simple implementation. 
We define the \emph{rate} of a PIR-CSI scheme as its minimum download rate (i.e., the inverse of the normalized download cost) amongst all problem instances, and define the \emph{server-symmetric capacity} of the PIR-CSI problem as the supremum of rates over all server-symmetric PIR-CSI schemes. 

For the settings in which the side information is not a function of the requested message, we show that the capacity is given by ${(1+{1}/{N}+\dots+{1}/{N^{\left\lceil \frac{K}{M+1}\right\rceil -1}})^{-1}}$ for any ${1\leq M\leq K-1}$. Interestingly, the capacity in this case is the same as that of multi-server PIR-SI under the server-symmetry assumption~\cite{li2018converse}, where the user's side information is comprised of $M$ uncoded messages. Moreover, comparing this result with the capacity of multi-server PIR without side information~\cite{sun2017capacity}, one can see that having a coded side information (which is not a function of the demanded message) of size $M$ reduces the effective number of messages from $K$ to $\lceil K/(M+1)\rceil$. 

For the settings wherein the side information is a function of the requested message, we show that the capacity is equal to $1$ for $M=2$ and $M=K$, and it is equal to $N/(N+1)$ for any $3 \leq M \leq K-1$. Again, a comparison of these results with the capacity of multi-server PIR without side information reveals that having a coded side information (which is a function of the demanded message) of size $M\in \{2,K\}$ and $M\in \{3,\dots,K-1\}$ reduces the effective number of messages from $K$ to $1$ and $2$, respectively. 

The converse proofs rely on new information-theoretic arguments, and the achievability schemes are inspired by our proposed scheme in~\cite{heidarzadeh2018capacity} for single-server PIR-CSI as well as the Sun-Jafar scheme of~\cite{sun2017capacity} for multi-server PIR.

\section{Problem Formulation}\label{sec:SN}
Throughout, we denote random variables by bold-face letters and their realizations by regular letters. 

Let $\mathbb{F}_q$ be a finite field of size $q$ for some prime $q$, and let $\mathbb{F}_q^{\times} \triangleq \mathbb{F}_q\setminus \{0\}$ be the multiplicative group of $\mathbb{F}_q$. Let $\mathbb{F}_{q^m}$ be an extension field of $\mathbb{F}_q$ for some integer $m\geq 1$, and let $L \triangleq m\log_2 q$. For an integer $i\geq 1$, let $[i]\triangleq \{1,\dots,i\}$. 

There are $N$ non-colluding servers, each storing an identical copy of $K$ messages $X_1,\dots,X_K$ where $\mathbf{X}_i$ for $i\in [K]$ is independently and uniformly distributed over $\mathbb{F}_{q^m}$, i.e., ${H(\mathbf{X}_i) = L}$ for all $i\in [K]$ and ${H(\mathbf{X}_1,\dots,\mathbf{X}_K) = KL}$. We denote ${X_{[K]}\triangleq\{X_1,\dots,X_K\}}$ and ${\mathbf{X}_{[K]}\triangleq\{\mathbf{X}_1,\dots,\mathbf{X}_K\}}$. There is a user who wishes to download a message $X_W$ for some ${W\in [K]}$ from the servers. We refer to $W$ as the \emph{demand index} and $X_W$ as the \emph{demand}. The user has a linear combination ${Y^{[S,C]}\triangleq \sum_{i\in S} c_i X_i}$ for some $S \triangleq \{i_1,\dots,i_M\}\subseteq [K]$ and ${C \triangleq \{c_{i_1},\dots,c_{i_M}\}}$ with $c_i\in \mathbb{F}^{\times}_q$. We refer to $Y^{[S,C]}$ as the \emph{side information}, $S$ as the \emph{side information index set}, and $M$ as the \emph{side information size}. 

Let $\mathcal{S}$ be the set of all $M$-subsets of $[K]$, and let $\mathcal{C}$ be the set of all length-$M$ sequences (i.e., ordered sets of size $M$) with elements from $\mathbb{F}^{\times}_q$. 
We assume that $\mathbf{S}$ and $\mathbf{C}$ are uniformly distributed over $\mathcal{S}$ and $\mathcal{C}$, respectively. We consider two different models for the conditional distribution of $\mathbf{W}$ given $\mathbf{S}$, depending on whether $\mathbf{W}\not\in \mathbf{S}$ (Model~I) or $\mathbf{W}\in \mathbf{S}$ (Model~II). For Model~I, we assume that $\mathbf{W}$ is distributed uniformly over $[K]\setminus S$ given that $\mathbf{S}=S$; and for Model~II, we assume that $\mathbf{W}$ is distributed uniformly over $S$ given that $\mathbf{S}=S$. 
Note that Models~I and~II are well-defined for $1\leq M\leq K-1$ and $2\leq M\leq K$, respectively.  

We assume that the servers initially know the side information size $M$, the considered model (i.e., whether $\mathbf{W}\not\in \mathbf{S}$ or $\mathbf{W}\in \mathbf{S}$), and the distributions of $\mathbf{S}$ and $\mathbf{C}$, and the conditional distribution of $\mathbf{W}$ given $\mathbf{S}$; whereas the servers have no information about the realizations $W,S,C$ \emph{a priori}. 

In order to retrieve $X_W$ for any given $W,S,C$, the user generates $N$ queries $Q_n^{[W,S,C]}$ for $n\in [N]$, and sends to the $n$th server the query $Q_n^{[W,S,C]}$, which is a (potentially stochastic) function of $W,S,C,Y^{[S,C]}$. 

Upon receiving $Q_n^{[W,S,C]}$, the $n$th server responds to the user with an answer $A_n^{[W,S,C]}$, which is a (deterministic) function of the query $Q_n^{[W,S,C]}$ and the messages in $X_{[K]}$. In particular, ${(\mathbf{W},\mathbf{S},\mathbf{C}) \rightarrow (\mathbf{Q}_n^{[\mathbf{W},\mathbf{S},\mathbf{C}]},\mathbf{X}_{[K]}) \rightarrow \mathbf{A}_n^{[\mathbf{W},\mathbf{S},\mathbf{C}]}}$ is a Markov chain, and ${H(\mathbf{A}_n^{[\mathbf{W},\mathbf{S},\mathbf{C}]}| \mathbf{Q}_n^{[\mathbf{W},\mathbf{S},\mathbf{C}]},\mathbf{X}_{[K]}) = 0}$. 

For the simplicity, we denote ${Q}^{[W,S,C]}\triangleq\{Q_n^{[W,S,C]}\}_{n\in [N]}$ and $A^{[W,S,C]}\triangleq\{A_n^{[W,S,C]}\}_{n\in [N]}$. 

\emph{\textbf{Recoverability condition}}:
The user must be able to retrieve the demand $X_W$ from the answers $A^{[W,S,C]}$ along with the side information $Y^{[S,C]}$, i.e.,
\[H(\mathbf{X}_{\mathbf{W}}| \mathbf{A}^{[\mathbf{W},\mathbf{S},\mathbf{C}]},\mathbf{Q}^{[\mathbf{W},\mathbf{S},\mathbf{C}]}, \mathbf{Y}^{[\mathbf{S},\mathbf{C}]},\mathbf{W},\mathbf{S},\mathbf{C})=0.\]

\emph{$\bm{W}$\textbf{-privacy condition}}:
For each $n\in [N]$, the query $Q^{[W,S,C]}_n$ must protect the privacy of the user's demand index $W$ from the $n$th server, i.e., for all $n\in [N]$,
\[I(\mathbf{W}; \mathbf{Q}_n^{[\mathbf{W},\mathbf{S},\mathbf{C}]},\mathbf{A}_n^{[\mathbf{W},\mathbf{S},\mathbf{C}]},\mathbf{X}_{[K]})=0.\] 


The problem is to design a protocol for generating queries $Q^{[W,S,C]}_n$ and their corresponding answers $A^{[W,S,C]}_n$ (for any given $W,S,C$) that satisfy the $W$-privacy and recoverability conditions. 
We refer to this problem as the \emph{Private Information Retrieval with Coded Side Information (PIR-CSI)}. In particular, we refer to this problem under Model~I (or Model~II) as the \emph{PIR-CSI--I} (or \emph{PIR-CSI--II}) \emph{problem}. 

We focus on server-symmetric PIR-CSI protocols in which the queries/answers to/from different servers are symmetric in structure. In particular, we say that a PIR-CSI--I (or respectively, PIR-CSI--II) protocol is \emph{server-symmetric} if  \[(\mathbf{Q}^{[W,S,C]}_{n},\mathbf{A}^{[W,S,C]}_{n},\mathbf{X}_{[K]})\sim (\mathbf{Q}^{[W,S,C]}_{n'},\mathbf{A}^{[W,S,C]}_{n'},\mathbf{X}_{[K]})\]  holds  for all $n,n'\in [N]$ and for any $W\in [K],S\in \mathcal{S},C\in \mathcal{C}$ such that ${W\not\in S}$ (or respectively, ${W\in S}$), where the relation $\mathbf{U}\sim \mathbf{V}$ means that $\mathbf{U}$ and $\mathbf{V}$ have identical distributions. It should be noted that most of the existing multi-server PIR schemes (with information-theoretic guarantees) are server-symmetric. Server-symmetric schemes are particularly of interest because the symmetry of queries/answers across the servers makes the implementation quite simple in practice. 

We define the \emph{rate} of a PIR-CSI--I (or PIR-CSI--II) protocol as the ratio of the entropy of a message, i.e., $L$, to the maximum total entropy of the answers from all servers, i.e., ${D\triangleq\max_{\{W,S,C\}} H(\mathbf{A}^{[W,S,C]})}$, where the maximization is over all $W,S,C$ such that $W\not\in S$ (or $W\in S$). We also define the \emph{server-symmetric capacity} of the PIR-CSI--I (or PIR-CSI--II) problem as the supremum of rates over all server-symmetric PIR-CSI--I (or PIR-CSI--II) protocols.

Our goal is to characterize the server-symmetric capacity of the PIR-CSI--I and PIR-CSI--II problems, and to design server-symmetric protocols that are capacity-achieving. 



\section{Main Results}

In this section, we present our main results. Theorems~\ref{thm:PIRCSI-I} and~\ref{thm:PIRCSI-II} characterize the server-symmetric capacity of the PIR-CSI--I and PIR-CSI--II problems, respectively. 

\begin{theorem}\label{thm:PIRCSI-I}
The server-symmetric capacity of PIR-CSI--I problem with $N$ servers, $K$ messages, and side information size $0\leq M \leq K-1$ is given by
\[{\mathrm{C}_{W-\text{\it I}}} = \left(1+\frac{1}{N}+\dots+\frac{1}{N^{\lceil \frac{K}{M+1}\rceil -1}}\right)^{-1}.\]
\end{theorem}

This result is interesting because it shows that the capacity in this case is the same as the capacity of multi-server PIR-SI~\cite{li2018converse,kadhe2017private} where $M$ uncoded messages are available at the user as side information. That is, 
knowing only one linear combination of $M$ messages as side information would be as effective as knowing $M$ (uncoded) messages separately.


\begin{theorem}\label{thm:PIRCSI-II}
The server-symmetric capacity of PIR-CSI--II problem with $N$ servers, $K$ messages, and side information size $2\leq M\leq K$ is given by

\[
{\mathrm{C}_{W-\text{II}}} = \begin{cases}
1, & \quad M=2,K,\\
\frac{N}{N+1}, & \quad 3 \leq M \leq K-1.
\end{cases} .
\]

\end{theorem}
 
This result shows that for the two corner cases of $M=2$ and $M=K$, the cost of retrieving one message privately is no more than that of downloading the message directly. For the cases of $3 \leq M \leq K-1$, full privacy can be achieved for only an additional download cost of $L/N$.  
\section{The ~PIR-CSI--I ~Problem}\label{sec:PIR-CSI-I}
In this section, we give the proof of converse and the achievability proof of Theorem~\ref{thm:PIRCSI-I}. 


The following lemma renders a necessary condition for any server-symmetric PIR-CSI--I (or PIR-CSI--II) protocol to satisfy the $W$-privacy condition. 

\begin{lemma}\label{privacy & Recoverability} 
Any server-symmetric \mbox{PIR-CSI--I} (or PIR-CSI--II) protocol satisfies the following condition: for any ${W,W'\in [K]}, {S\in \mathcal{S}}, {C \in \mathcal{C}}$ with ${W\not\in S}$ (or ${W\in S}$), there exist ${S'\in \mathcal{S}}, {C'\in \mathcal{C}}$ with ${W'\not\in S'}$ (or ${W'\in S'}$), such that \[(\mathbf{Q}_n^{[W,S,C]},\hspace{-0.05cm}\mathbf{A}_n^{[W,S,C]},\mathbf{X}_{[K]})\hspace{-0.1cm}\sim\hspace{-0.1cm} 
(\mathbf{Q}_n^{[W',S',C']},\hspace{-0.05cm}\mathbf{A}_n^{[W',S',C']},\mathbf{X}_{[K]})\] holds for all $n\in [N]$. 
\end{lemma}

\begin{proof}
The proof is by the way of contradiction, and based on the definitions of \mbox{$W$-privacy} and server-symmetry. To protect the user's privacy, for different demands, the strategies (queries and answers) must be indistinguishable (identically distributed) from the perspective of each server. In particular, for each $n\in [N]$, it must hold that for any ${W\in [K]},{S\in \mathcal{S}}, {C\in \mathcal{C}}$ with ${W\not\in S}$ (or ${W\in S}$), and any candidate demand ${W' \in [K]}$, there exist ${S_n\in \mathcal{S},C_n\in \mathcal{C}}$ with ${W'\not\in S_n}$ (or ${W'\in S_n}$) that satisfy the condition of the lemma for the server $n$. 
Otherwise, if there do not exist such ${S_n\in \mathcal{S},C_n\in \mathcal{C}}$ that satisfy the condition of the lemma for some server $n$, then the privacy condition is violated. Moreover, by the server-symmetry assumption, for any candidate demand ${W' \in [K]}$, there must exist ${S'\in \mathcal{S},C'\in \mathcal{C}}$ (independent of $n$) with ${W'\not\in S'}$ (or $W'\in S'$) that make the strategies indistinguishable from the perspective of each server $n\in [N]$. That is, there must exist $S'\in \mathcal{S},C'\in \mathcal{C}$ with $W'\not\in S'$ (or $W'\in S'$) that satisfy the condition of the lemma for all servers. Otherwise, the server-symmetry assumption is violated. 
\end{proof}


\subsection{Converse Proof for Theorem~\ref{thm:PIRCSI-I}}
Suppose that the user wishes to retrieve $X_W$ for a given $W \in [K]$, and has a side information ${Y \triangleq Y^{[S,C]}}$ for given  $S\in \mathcal{S},C\in \mathcal{C}$ such that $W \not \in S$. The user sends to the $n$th server a query $Q_n^{[W,S,C]}$, and the $n$th server responds to the user with an answer $A_n^{[W,S,C]}$. We need to show that the maximum total entropy of the answers from all servers (over all $W,S,C$), denoted by $D$, is lower bounded by ${(1+{1}/{N}+\dots+{1}/{N^{\lceil \frac{K}{M+1}\rceil-1}})L}$. 

The proof proceeds as follows: 
\begin{align}
D &\geq H(\mathbf{A}^{[W,S,C]}|\mathbf{Q}^{[W,S,C]},\mathbf{Y})\nonumber \\ &=H(\mathbf{A}^{[W,S,C]},\mathbf{X}_W|\mathbf{Q}^{[W,S,C]},\mathbf{Y})\label{eq:line1}\\
&=L+H(\mathbf{A}^{[W,S,C]}|\mathbf{Q}^{[W,S,C]},\mathbf{X}_W,\mathbf{Y})\label{eq:line2}\\
&\geq L+H({\mathbf{A}_1^{[W,S,C]}}|{\mathbf{Q}^{[W,S,C]}},\mathbf{X}_W,\mathbf{Y})\nonumber\\
&= L+H({\mathbf{A}_1^{[W,S,C]}}|{\mathbf{Q}_1^{[W,S,C]}},\mathbf{X}_W,\mathbf{Y})\label{eq:line3}
\end{align}
where~\eqref{eq:line1} follows from ${H(\mathbf{X}_W|{\mathbf{A}^{[W,S,C]}},{\mathbf{Q}^{[W,S,C]}},\mathbf{Y})=0}$ (by the recoverability condition);~\eqref{eq:line2} holds since $\mathbf{X}_W$ is independent of $({\mathbf{Q}^{[W,S,C]}},\mathbf{Y})$, and $H(\mathbf{X}_W|{\mathbf{Q}^{[W,S,C]}},\mathbf{Y})=H(\mathbf{X}_W)=L$; and~\eqref{eq:line3} holds because $\mathbf{A}_1^{[W,S,C]}$ only depends on ${({\mathbf{Q}_1^{[W,S,C]}},\mathbf{X}_{[K]})}$, and is conditionally independent of ${\mathbf{Q}_{n}^{[W,S,C]}}$ for all $n\neq 1$, given $({\mathbf{Q}_1^{[W,S,C]}},\mathbf{X}_W,\mathbf{Y})$. 

We will consider the following two cases separately: (i) ${\lceil K/(M+1) \rceil =1}$ (i.e., $K=M+1$), and (ii) ${\lceil K/(M+1) \rceil >1}$ (i.e., $K>M+1$). In the case (i), we need to show that $D$ is lower bounded by $L$. Since $H({\mathbf{A}_1^{[W,S,C]}}|{\mathbf{Q}_1^{[W,S,C]}},\mathbf{X}_W,\mathbf{Y})\label{eq:line4} \geq 0$, then $D \geq L$ (by~\eqref{eq:line3}). 

In the case (ii), in order to continue lower bounding~\eqref{eq:line3}, we arbitrarily choose a message, say $X_{W_1}$, such that $W_1 \not\in W \cup S$. (Note that such $W_1$ exists because $|W\cup S|=M+1<K$.) Based on Lemma~\ref{privacy & Recoverability}, there exist
$S_1\in \mathcal{S}$, $C_1 \in \mathcal{C}$ with $W_1 \not \in S_1$, and accordingly ${Y_1 \triangleq Y^{[S_1,C_1]}}$, such that \[H(\mathbf{A}_1^{[W,S,C]}|\mathbf{Q}_1^{[W,S,C]},\mathbf{X}_W,\mathbf{Y}) =H(\mathbf{A}_1^{[W_1,S_1,C_1]}|\mathbf{Q}_1^{[W_1,S_1,C_1]},\mathbf{X}_W,\mathbf{Y}).\] Then, we can write 
\begin{align}
D &\geq L+H(\mathbf{A}_1^{[W,S,C]}|\mathbf{Q}_1^{[W,S,C]},\mathbf{X}_W,\mathbf{Y})\nonumber\\ 
&= L+H({\mathbf{A}_1^{[W_1,S_1,C_1]}}|{\mathbf{Q}_1^{[W_1,S_1,C_1]}},\mathbf{X}_W,\mathbf{Y})\nonumber\\
&\geq L+H({\mathbf{A}_1^{[W_1,S_1,C_1]}}|{\mathbf{Q}^{[W_1,S_1,C_1]}},\mathbf{X}_W,\mathbf{Y}).\nonumber
\end{align} 
Similarly, by the server-symmetry assumption we have
\begin{align}
D &\geq L+H({\mathbf{A}_n^{[W_1,S_1,C_1]}}|{\mathbf{Q}^{[W_1,S_1,C_1]}},\mathbf{X}_W,\mathbf{Y})\nonumber
\end{align} for all $n \in [N]$. Combining all of these inequalities, we get 
\begin{align}
D &\geq L+\frac{1}{N}\sum_{n=1}^{N} H({\mathbf{A}_n^{[W_1,S_1,C_1]}}|{\mathbf{Q}^{[W_1,S_1,C_1]}},\mathbf{X}_W,\mathbf{Y})\nonumber \\
& \geq L+\frac{1}{N}H({\mathbf{A}^{[W_1,S_1,C_1]}}|{\mathbf{Q}^{[W_1,S_1,C_1]}},\mathbf{X}_W,\mathbf{Y}).\label{eq:line4new}
\end{align} 

To further lower bound~\eqref{eq:line4new}, we can write
\begin{align}
&\hspace{-0.25cm}H({\mathbf{A}^{[W_1,S_1,C_1]}}|{\mathbf{Q}^{[W_1,S_1,C_1]}},\mathbf{X}_W,\mathbf{Y})\nonumber\\
&\hspace{-0.25cm}\geq H({\mathbf{A}^{[W_1,S_1,C_1]}}|{\mathbf{Q}^{[W_1,S_1,C_1]}},\mathbf{X}_W,\mathbf{Y},\mathbf{Y}_1)\nonumber\\ 
&\hspace{-0.25cm}= H({\mathbf{A}^{[W_1,S_1,C_1]}},\mathbf{X}_{W_1}|{\mathbf{Q}^{[W_1,S_1,C_1]}},\mathbf{X}_W,\mathbf{Y},\mathbf{Y}_1)\label{eq:line12}\\
&\hspace{-0.25cm}=L+H({\mathbf{A}^{[W_1,S_1,C_1]}}|{\mathbf{Q}^{[W_1,S_1,C_1]}},\mathbf{X}_W,\mathbf{Y},\mathbf{X}_{W_1},\mathbf{Y}_1)\label{eq:line13}\\
&\hspace{-0.25cm}\geq L+H({\mathbf{A}_1^{[W_1,S_1,C_1]}}|{\mathbf{Q}^{[W_1,S_1,C_1]}},\mathbf{X}_W,\mathbf{Y},\mathbf{X}_{W_1},\mathbf{Y}_1)\nonumber\\
&\hspace{-0.25cm}= L+H({\mathbf{A}_1^{[W_1,S_1,C_1]}}|{\mathbf{Q}_1^{[W_1,S_1,C_1]}},\mathbf{X}_W,\mathbf{Y},\mathbf{X}_{W_1},\mathbf{Y}_1)\label{eq:line7new}
\end{align} where~\eqref{eq:line12} holds since $X_{W_1}$ is recoverable from ${A^{[W_1,S_1 ,C_1]},Q^{[W_1,S_1 ,C_1]},Y_1,W_1,S_1,C_1}$; and~\eqref{eq:line13} holds because $\mathbf{X}_{W_1}$ is independent of $({\mathbf{Q}^{[W_1,S_1,C_1]}},\mathbf{X}_W,\mathbf{Y},\mathbf{Y}_1)$. 

We consider two cases separately: (ii.1) $\lceil K/(M+1)\rceil = 2$, and (ii.2) $\lceil K/(M+1)\rceil > 2$. In the case (ii.1), from~\eqref{eq:line4new} and~\eqref{eq:line7new} it follows that $D\geq L+L/N$.

In the case (ii.2), to continue lower bounding~\eqref{eq:line7new}, we pick a message, say ${X_{W_2}}$, such that $W_2 \not\in W\cup S \cup W_1 \cup S_1$. (Note that such $W_2$ exists since $|W\cup S\cup W_1\cup S_1|\leq 2(M+1)<K$.) According to Lemma \ref{privacy & Recoverability}, there exist
$S_2\in \mathcal{S}$, $C_2 \in \mathcal{C}$ with ${W_2 \not \in S_2}$, and accordingly, ${Y_2 = Y^{[S_2,C_2]}}$, such that \[H(\mathbf{A}_1^{[W_1,S_1,C_1]}|\mathbf{Q}_1^{[W_1,S_1,C_1]},\mathbf{X}_W,\mathbf{Y},\mathbf{X}_{W_1},\mathbf{Y}_1) = H(\mathbf{A}_1^{[W_2,S_2,C_2]}|\mathbf{Q}_1^{[W_2,S_2,C_2]},\mathbf{X}_W,\mathbf{Y},\mathbf{X}_{W_1},\mathbf{Y}_1).\] Thus,
\begin{align}
&H({\mathbf{A}^{[W_1,S_1,C_1]}}|{\mathbf{Q}^{[W_1,S_1,C_1]}},\mathbf{X}_W,\mathbf{Y})\nonumber\\
& \geq L+H({\mathbf{A}_1^{[W_1,S_1,C_1]}}|{\mathbf{Q}_1^{[W_1,S_1,C_1]}},\mathbf{X}_W,\mathbf{Y},\mathbf{X}_{W_1},\mathbf{Y}_1)\nonumber\\
& = L+H({\mathbf{A}_1^{[W_2,S_2,C_2]}}|{\mathbf{Q}_1^{[W_2,S_2,C_2]}},\mathbf{X}_W,\mathbf{Y},\mathbf{X}_{W_1},\mathbf{Y}_1)\nonumber\\
&\geq L+H({\mathbf{A}_1^{[W_2,S_2,C_2]}}|{\mathbf{Q}^{[W_2,S_2,C_2]}},\mathbf{X}_W,\mathbf{Y},\mathbf{X}_{W_1},\mathbf{Y}_1).\nonumber
\end{align} Similarly, by the server-symmetry assumption, we have
\begin{align}
&H({\mathbf{A}^{[W_1,S_1,C_1]}}|{\mathbf{Q}^{[W_1,S_1,C_1]}},\mathbf{X}_W,\mathbf{Y}) \nonumber\\
& \geq L+H({\mathbf{A}_n^{[W_2,S_2,C_2]}}|{\mathbf{Q}^{[W_2,S_2,C_2]}},\mathbf{X}_W,\mathbf{Y},\mathbf{X}_{W_1},\mathbf{Y}_1)\nonumber
\end{align} for all $n \in [N]$. Combining all of these inequalities, we get
\begin{align}
&H({\mathbf{A}^{[W_1,S_1,C_1]}}|{\mathbf{Q}^{[W_1,S_1,C_1]}},\mathbf{X}_W,\mathbf{Y})  \nonumber\\
& \quad \geq L+ \frac{1}{N}H({\mathbf{A}^{[W_2,S_2,C_2]}}|{\mathbf{Q}^{[W_2,S_2,C_2]}},\mathbf{X}_W,\mathbf{Y},\mathbf{X}_{W_1},\mathbf{Y}_1).\label{eq:line8new}
\end{align} Putting~\eqref{eq:line4new} and~\eqref{eq:line8new} together, we get

\begin{dmath*}
D \geq L+\frac{L}{N} + \frac{1}{N^2}H({\mathbf{A}^{[W_2,S_2,C_2]}}|{\mathbf{Q}^{[W_2,S_2,C_2]}},\mathbf{X}_W,\mathbf{Y},\mathbf{X}_{W_1},\mathbf{Y}_1).
\end{dmath*} 

By recursively choosing the messages $X_{W_i}$ (similarly as $X_{W_1}$ and $X_{W_2}$) for $i\in \{3,\dots,{\lceil K/(M+1)\rceil}\}$ and using the same lower bounding technique, it can be shown that
\begin{align}
D &\geq  L+{L}/{N}+\dots+{L}/{N^{\lceil \frac{K}{M+1}\rceil-1}}.\nonumber
\end{align} 


\subsection{Achievability Proof for Theorem~\ref{thm:PIRCSI-I}}\label{sec:PIRCSIIpro}
In this section, we propose a server-symmetric PIR-CSI--I protocol that achieves a rate equal to ${\mathrm{C}_{W-\text{\it I}}}$. The proposed protocol employs the \emph{Randomized Partitioning (RP)} scheme which we proposed in~\cite{heidarzadeh2018capacity} for single-server PIR-CSI (under Model~I) as well as the Sun-Jafar scheme of~\cite{sun2017capacity} for multi-server PIR. 

We assume that each message consists of ${N^{\lceil K/(M+1) \rceil}}$ symbols over $\mathbb{F}_q$. 
\vspace{0.25cm}  

\textbf{Multi-Server PIR-CSI--I Protocol:} 
\vspace{0.125cm} 


\textit{\textbf{Step 1:}} The user utilizes the RP scheme of~\cite{heidarzadeh2018capacity} to construct $r\triangleq\lceil \frac{K}{M+1}\rceil$ sequences $I_1,\dots,I_r$ from indices in $[K]$, each of length $M+1$, and $r$ sequences $I'_1,\dots,I'_r$ with elements in $\mathbb{F}^{\times}_q$, each of length $M+1$. In particular, $I_1 = \{W,S\}$ and $I'_1 = \{c,C\}$ where $C$ is the sequence of coefficients in the user's side information $Y^{[S,C]}$, and $c$ is randomly chosen from $\mathbb{F}^{\times}_q$. (For more details, see~\cite[Section IV-B]{heidarzadeh2018capacity}.)

\textit{\textbf{Step 2:}} 
The user then creates $\tilde{I}_i$ and $\tilde{I}'_i$ for each $i\in [r]$ by reordering the elements of both $I_i$ and $I'_i$ with the same randomly picked permutation $\pi_i: [M+1]\rightarrow [M+1]$, and constructs $I^{*}_i = (\tilde{I}_i,\tilde{I}'_i)$. Then, the user sends $\{I^{*}_{\sigma(i)}\}_{i\in [r]}$ to all servers, 
for a randomly chosen permutation ${\sigma: [r]\rightarrow [r]}$. Note that in the RP scheme, $\{I_i\}_{i\in [r]}$ and $\{I'_i\}_{i\in [r]}$ are designed in such a way that given $\{I^{*}_{\sigma(1)},\dots,I^{*}_{\sigma(r)}\}$, any index in $[K]$ is equally likely to be the user's demand index.


\textit{\textbf{Step 3:}} Using $I^{*}_{\sigma(i)} = (\tilde{I}_{\sigma(i)},\tilde{I}'_{\sigma(i)})$ for all $i\in [r]$, the user and all the servers form $r$ \textit{super-messages} ${\hat{X}_1,\dots,\hat{X}_r}$ such that $\hat{X}_{i}=\sum_{j=1}^{M+1} c_{i_j} X_{i_j}$ for all $i\in [r]$, where ${\tilde{I}_{\sigma(i)} = \{i_1,\dots,i_{M+1}\}}$ and ${\tilde{I}'_{\sigma(i)} = \{c_{i_1},\dots,c_{i_{M+1}}\}}$. 

\textit{\textbf{Step 4:}} The user and the servers then utilize the Sun-Jafar protocol with $r$ super-messages $\hat{X}_1,\dots,\hat{X}_r$ in such a way that the user can privately download the super-message $\hat{X}_{\sigma^{-1}(1)}=cX_W+Y^{[S,C]}$; and subsequently, subtracting off $Y^{[S,C]}$ from $\hat{X}_{\sigma^{-1}(1)}$, the user recovers $X_W$. 

\textbf{Example 1.} Assume that there are ${N = 2}$ servers, ${K=9}$ messages from $\mathbb{F}_{3^8}$ (i.e., each message has $8$ symbols over $\mathbb{F}_3$), and $M=3$. Suppose that the user demands the message $X_1$ and has a side information $X_2+2X_3+X_4$. Note that for this example, $W = 1$, $S= \{2,3,4\}$, and $C=\{1,2,1\}$. 

First, the user labels $r=\lceil \frac{K}{M+1}\rceil=3$ sequences as $I_1,I_2,I_3$, each of length $M+1=4$. For creating these sequences, the user needs to have $12$ indices, but at the beginning the user has $9$ indices. For selecting the remaining $3$ required indices, following the RP scheme of~\cite{heidarzadeh2018capacity}, the user selects $w\in \{0,1\}$, $s\in \{0,1,2,3\}$, and $t\in\{0,1,\dots,5\}$ randomly chosen indices from $W=\{1\}$, $S = \{2,3,4\}$, and $T = \{5,6,7,8,9\}$, respectively, according to a carefully designed probability distribution (ensuring $W$-privacy of the RP scheme) on all $(w,s,t)$ such that $w+s+t= 3$. For this example, the probability distribution is given by \vspace{-0.125cm} 
\begin{align*}
p(w,s,t) \triangleq 
\begin{cases}
\frac{14}{171}, & w=0,s=3,t=0\\
\frac{60}{171}, & w=0,s=2,t=1\\
\frac{36}{171}, & w=0,s=1,t=2\\
\frac{4}{171}, & w=0,s=0,t=3\\
\frac{21}{171}, & w=1,s=2,t=0\\
\frac{30}{171}, & w=1,s=1,t=1\\
\frac{6}{171}, & w=1,s=0,t=2
\end{cases}  
\end{align*} \vspace{-0.125cm}

Suppose that the user chooses ${w=1},{s=1},{t=1}$, and selects the $3$ indices $\{1,2,5\}$. Following the RP protocol, the user forms the sequence $I_1=\{W,S\}=\{1,2,3,4\}$. In the remaining $8$ indices, there is one repetitive index, $5$. For forming the other two sequences, $I_2$ and $I_3$, the user places the repetitive index $5$ into both $I_2$ and $I_3$. Next, the user randomly partitions the remaining $6$ indices, $\{1,2,6,7,8,9\}$, into $I_2$ and $I_3$. For this example, suppose that ${I_2=\{5,1,7,8\}}$ and ${I_3=\{5,2,6,9\}}$. 

The user then labels ${r=3}$ sequences as $I'_1,I'_2,I'_3$, each of length $4$. For this example, suppose that the user creates ${I'_1=I'_2=I'_3=\{1,1,2,1\}}$. Then, the user randomly reorders the elements of $I_i$ and $I'_i$, and constructs\vspace{-0.1cm} 
\begin{align*}
\tilde{I}_1 = \{2,4,1,3\}, &\quad  \tilde{I}'_1 = \{1,1,1,2\}\\ 
\tilde{I}_2 = \{7,5,1,8\}, &\quad \tilde{I}'_2 = \{2,1,1,1\}\\
\tilde{I}_3 = \{2,9,6,5\}, &\quad \tilde{I}'_3 = \{1,1,2,1\}.
\end{align*}

\vspace{-0.1cm}
Next, the user sends a uniform random permutation of ${\{I^{*}_1,I^{*}_2,I^{*}_3\}}$, say ${\{I^{*}_1,I^{*}_3,I^{*}_2\}}$, to both servers, where ${I^{*}_i = (\tilde{I}_i,\tilde{I}'_i)}$. 

The user and the servers then form three super-messages as follows:\vspace{-0.125cm} 
\begin{align*}
\hat{X}_1 &=X_2+X_4+X_1+2X_3\\
\hat{X}_2 &=X_2+X_9+2X_6+X_5\\
\hat{X}_3 &=2X_7+X_5+X_1+X_8.
\end{align*}

\vspace{-0.125cm}
Finally, the user and the servers run the Sun-Jafar protocol as follows for the three super-messages $\hat{X}_1,\hat{X}_2,\hat{X}_3$ in such a way that the user can privately download $\hat{X}_1$. For each $\hat{X}_i$, let $[\hat{X}_{i,1},\dots,\hat{X}_{i,8}]$ be an independent and uniform random permutation of the $8$ symbols (over $\mathbb{F}_3$) of $\hat{X}_i$. The user requests $7$ symbols from the first server and $7$ symbols from the second server as listed in Table~\ref{table:1}~\cite{sun2017capacity}, where the requested symbols are carefully designed linear combinations of symbols $\{\hat{X}_{i,j}\}_{i\in [3],j\in [8]}$.  
From the servers' answers, the user first obtains the super-message $\hat{X}_1=X_2+X_4+X_1+2X_3$, and then recovers the desired message $X_1$ by subtracting off the side information $X_2+2X_3+X_4$. 
For this example, the proposed protocol requires to download a total of $14$ symbols (over $\mathbb{F}_3$), achieving the rate of $8/14=4/7$.

\begin{lemma}\label{lemma 2}
The Multi-Server PIR-CSI--I protocol is a server-symmetric protocol that satisfies the \mbox{recoverability} and the $W$-privacy conditions, and achieves the rate ${(1+{1}/{N}+\dots+{1}/{N^{\lceil \frac{K}{M+1}\rceil -1}})^{-1}}$.\vspace{-0.5cm}
\end{lemma}

\begin{table}[t!]
\caption{The queries/answers of Sun-Jafar protocol for $2$ servers and $3$ messages $\hat{X}_1,\hat{X}_2,\hat{X_3}$, when the user demands $\hat{X}_{1}$\cite{sun2017capacity}.}\vspace{-0.125cm}
    \centering
    \scalebox{1.15}{
\begin{tabular}{ |c|c| } 
 \hline
 Server 1 & Server 2 \\
 \hline
 $\hat{X}_{1,1},\hat{X}_{2,1},\hat{X}_{3,1}$ & $\hat{X}_{1,2},\hat{X}_{2,2},\hat{X}_{3,2}$ \\ [0.5ex]
 $\hat{X}_{1,3}+\hat{X}_{2,2}$ & $\hat{X}_{1,5}+\hat{X}_{2,1}$ \\ [0.5ex]
 $\hat{X}_{1,4}+\hat{X}_{3,2}$ & $\hat{X}_{1,6}+\hat{X}_{3,1}$ \\ [0.5ex]
 $\hat{X}_{2,3}+\hat{X}_{3,3}$ & $\hat{X}_{2,4}+\hat{X}_{3,4}$ \\ [0.5ex]
 $\hat{X}_{1,7}+\hat{X}_{2,4}+\hat{X}_{3,4}$ & $\hat{X}_{1,8}+\hat{X}_{2,3}+\hat{X}_{3,3}$ \\ [0.5ex]
 \hline
\end{tabular}}\vspace{-0.35cm}
\label{table:1}
\end{table}

\begin{proof}

The Multi-Server PIR-CSI--I protocol is a server-symmetric protocol since as explained in Step~$4$ of this protocol, it builds upon the Sun-Jafar protocol that enforces symmetry across servers \cite{sun2017capacity}.

Since $\mathbf{X}_1,\dots,\mathbf{X}_K$ are uniformly and independently distributed over $\mathbb{F}_{q^m}$, and $\hat{X}_1,\dots,\hat{X}_r$ are linearly independent combinations of $X_1,\dots,X_K$ over $\mathbb{F}_q$, then $\hat{\mathbf{X}}_1,\dots,\hat{\mathbf{X}}_r$ are uniformly and independently distributed over $\mathbb{F}_{q^m}$. That is, ${H(\hat{\mathbf{X}}_i) = m\log_2 q=L}$ for all $i\in [r]$. Hence, the proposed protocol achieves the same rate as the Sun-Jafar protocol for $N$ servers and $\lceil {K}/{(M+1)}\rceil$ identically and independently distributed messages, i.e., the rate 
${(1+{1}/{N}+\dots+{1}/{N^{\lceil \frac{K}{M+1}\rceil-1}})^{-1}}$ (see~\cite[Theorem~1]{sun2017capacity}). 

From the step~$4$ of the proposed protocol, it can be easily confirmed that the recoverability condition is satisfied. The proof of $W$-privacy is as follows. By the design of the protocol, all servers are fully aware of how the super-messages $\hat{X}_1,\dots,\hat{X}_r$ have been formed. From the perspective of each server, according to the RP protocol, each super-message $\hat{X}_i$ has a certain probability to be the super-message needed by the user, i.e., the super-message from which the user can recover the demanded message. On the other hand, the Sun-Jafar protocol guarantees that given their query, no server can obtain any information about which super-message is being requested by the user. That is, given their query, from each server's perspective the probability of any super-message $\hat{X}_i$ to be the super-message needed by the user remains the same as that in the RP protocol. Moreover, the $W$-privacy of the RP protocol ensures that given their query, each server finds every message in $X_{[K]}$ equally likely to be the user's demand. This proves the $W$-privacy of the proposed protocol. 
\end{proof}

\section{The ~PIR-CSI--II ~Problem}\label{sec:PIR-CSI-II}

In this section, we give the proof of converse and the achievability proof of Theorem~\ref{thm:PIRCSI-II}. 

\subsection{Converse Proof for Theorem~\ref{thm:PIRCSI-II}}
Suppose that the user wishes to retrieve $X_W$ for a given $W \in [K]$, and has a side information ${Y \triangleq Y^{[S,C]}}$ for given $S\in \mathcal{S},C\in \mathcal{C}$ such that $W \in S$. We need to show that the maximum total entropy of the answers from all servers (over all $W,S,C$), denoted by $D$, is lower bounded by $L$ when $M=2$ or $M=K$, and is lower bounded by $(1+1/N)L$ when $3\leq M\leq K-1$. 

The proof proceeds as follows:  
\begin{align}
D &\geq H({\mathbf{A}^{[W,S,C]}}|{\mathbf{Q}^{[W,S,C]}},\mathbf{Y})\nonumber \\ &=H({\mathbf{A}^{[W,S,C]}},\mathbf{X}_W|{\mathbf{Q}^{[W,S,C]}},\mathbf{Y})\label{eq:line22}\\
&=L+H({\mathbf{A}^{[W,S,C]}}|{\mathbf{Q}^{[W,S,C]}},\mathbf{X}_W,\mathbf{Y})\label{eq:line23}
\end{align}
where~\eqref{eq:line22} holds because of the recoverability condition, and~\eqref{eq:line23} holds because $\mathbf{X}_W$ is independent of $({\mathbf{Q}^{[W,S,C]}},\mathbf{Y})$. By the non-negativity of the entropy,~\eqref{eq:line23} yields $D \geq L$, which completes the proof for the cases of $M=2$ and $M=K$. For the cases of ${3\leq M\leq K-1}$, we continue lower bounding~\eqref{eq:line23} as follows:
\begin{align}
D &\geq L+H({\mathbf{A}^{[W,S,C]}}|{\mathbf{Q}^{[W,S,C]}},\mathbf{X}_W,\mathbf{Y})\nonumber\\
&\geq L+H({\mathbf{A}_1^{[W,S,C]}}|{\mathbf{Q}^{[W,S,C]}},\mathbf{X}_W,\mathbf{Y})\nonumber\\
&= L+H({\mathbf{A}_1^{[W,S,C]}}|{\mathbf{Q}_1^{[W,S,C]}},\mathbf{X}_W,\mathbf{Y})\label{eq:line29}
\end{align}
where~\eqref{eq:line29} holds because given $({\mathbf{Q}_1^{[W,S,C]}},\mathbf{X}_W,\mathbf{Y})$, ${\mathbf{A}_1^{[W,S,C]}}$ is conditionally independent of $\mathbf{Q}_n^{[W,S,C]}$ for all ${n\neq 1}$. In order to continue lower bounding~\eqref{eq:line29}, we choose an arbitrary message, say $X_{W_1}$, such that $W_1 \in S\setminus W$. According to Lemma~\ref{privacy & Recoverability}, there exist $S_1\in \mathcal{S}$, $C_1\in \mathcal{C}$ with $W_1\in S_1$, and accordingly $Y_1 \triangleq Y^{[S_1,C_1]}$, such that \[H(\mathbf{A}_1^{[W,S,C]}|\mathbf{Q}_1^{[W,S,C]},\mathbf{X}_W,\mathbf{Y})= H(\mathbf{A}_1^{[W_1,S_1,C_1]}|\mathbf{Q}_1^{[W_1,S_1,C_1]},\mathbf{X}_W,\mathbf{Y}).\]
Rewriting~\eqref{eq:line29}, 
\begin{align}
D & \geq L+H({\mathbf{A}_1^{[W_1,S_1,C_1]}}|{\mathbf{Q}_1^{[W_1,S_1,C_1]}},\mathbf{X}_W,\mathbf{Y})\nonumber\\
&\geq L+H({\mathbf{A}_1^{[W_1,S_1,C_1]}}|{\mathbf{Q}^{[W_1,S_1,C_1]}},\mathbf{X}_W,\mathbf{Y}).\nonumber
\end{align} Similarly, by the server-symmetry assumption, we can write
\begin{align}
D &\geq L+H({\mathbf{A}_n^{[W_1,S_1,C_1]}}|{\mathbf{Q}^{[W_1,S_1,C_1]}},\mathbf{X}_W,\mathbf{Y})\nonumber
\end{align} for all $n \in [N]$. Combining all of these inequalities, we get
\begin{align}
D &\geq L+\frac{1}{N}H({\mathbf{A}^{[W_1,S_1,C_1]}}|{\mathbf{Q}^{[W_1,S_1,C_1]}},\mathbf{X}_W,\mathbf{Y}).\label{eq:line13new}
\end{align} To further lower bound $H({\mathbf{A}^{[W_1,S_1,C_1]}}|{\mathbf{Q}^{[W_1,S_1,C_1]}},\mathbf{X}_W,\mathbf{Y})$, we can write
\begin{align}
&H({\mathbf{A}^{[W_1,S_1,C_1]}}|{\mathbf{Q}^{[W_1,S_1,C_1]}},\mathbf{X}_W,\mathbf{Y})\nonumber\\ &\geq H({\mathbf{A}^{[W_1,S_1,C_1]}}|{\mathbf{Q}^{[W_1,S_1,C_1]}},\mathbf{X}_W,\mathbf{Y},\mathbf{Y}_1)\nonumber\\ &= H({\mathbf{A}^{[W_1,S_1,C_1]}},\mathbf{X}_{W_1}|{\mathbf{Q}^{[W_1,S_1,C_1]}},\mathbf{X}_W,\mathbf{Y},\mathbf{Y}_1)\label{eq:line34}
\end{align} where~\eqref{eq:line34} follows from the fact that $X_{W_1}$ is recoverable from $A^{[W_1,S_1,C_1]},Q^{[W_1,S_1,C_1]},Y_1,W_1,S_1,C_1$. 

We will consider two cases as follows separately: (i) $\mathbf{X}_{W_1}$ is independent of $({\mathbf{Q}^{[W_1,S_1,C_1]}},\mathbf{X}_W,\mathbf{Y},\mathbf{Y}_1)$, and (ii) $\mathbf{X}_{W_1}$ and $({\mathbf{Q}^{[W_1,S_1,C_1]}},\mathbf{X}_W,\mathbf{Y},\mathbf{Y}_1)$ are not independent. 

In the case (i), $\mathbf{X}_{W_1}$ and $({\mathbf{Q}^{[W_1,S_1,C_1]}},\mathbf{X}_W,\mathbf{Y},\mathbf{Y}_1)$ are independent. That is, $H(\mathbf{X}_{W_1}|\mathbf{Q}^{[W_1,S_1,C_1]},\mathbf{X}_W,\mathbf{Y},\mathbf{Y}_1) = H(\mathbf{X}_{W_1}) = L$. Then, we can continue lower bounding~\eqref{eq:line34} as follows:
\begin{align}
&H({\mathbf{A}^{[W_1,S_1,C_1]}}|{\mathbf{Q}^{[W_1,S_1,C_1]}},\mathbf{X}_W,\mathbf{Y})\nonumber\\
&\geq H({\mathbf{A}^{[W_1,S_1,C_1]}},\mathbf{X}_{W_1}|{\mathbf{Q}^{[W_1,S_1,C_1]}},\mathbf{X}_W,\mathbf{Y},\mathbf{Y}_1)\nonumber\\
&=H(\mathbf{X}_{W_1}|\mathbf{Q}^{[W_1,S_1,C_1]},\mathbf{X}_W,\mathbf{Y},\mathbf{Y}_1)\nonumber\\
&\quad + H(\mathbf{A}^{[W_1,S_1,C_1]}|\mathbf{Q}^{[W_1,S_1,C_1]},\mathbf{X}_W,\mathbf{Y},\mathbf{X}_{W_1},\mathbf{Y}_1)\nonumber\\
& = L +H(\mathbf{A}^{[W_1,S_1,C_1]}|\mathbf{Q}^{[W_1,S_1,C_1]},\mathbf{X}_W,\mathbf{Y},\mathbf{X}_{W_1},\mathbf{Y}_1)\nonumber\\
&\geq L.\label{eq:line16new}
\end{align} 

By~\eqref{eq:line13new} and~\eqref{eq:line16new}, $D\geq L+L/N$, as was to be shown. 

In the case (ii), due to the dependence of $\mathbf{X}_{W_1}$ and $({\mathbf{Q}^{[W_1,S_1,C_1]}},\mathbf{X}_W,\mathbf{Y},\mathbf{Y}_1)$ and the linearity of $X_W,Y,X_{W_1},Y_1$, it must hold that ${Y=c_W X_W+c_{W_1} X_{W_1}+Z}$ and ${Y_1=c'_W X_W + c'_{W_1} X_{W_1}+c'Z}$ for some ${c'_W,c'_{W_1},c'}\in \mathbb{F}^{\times}_q$, where $Z = \sum_{i\in S\setminus \{W,W_1\}} c_{i} X_i$, and ${c_i}$'s are the elements in the sequence $C$ (i.e., the coefficients of the messages in the side information $Y$). We proceed by lower bounding~\eqref{eq:line29}, when ${W,S,C,Y}$ are replaced by ${W_1,S_1,C_1,Y_1}$. To this end, we choose an arbitrary message, say $X_{W_2}$, such that ${{W_2}\not\in S}$. Based on Lemma~\ref{privacy & Recoverability}, there exist ${S_2\in \mathcal{S}}$, ${C_2\in \mathcal{C}}$ with $W_2\in S_2$, and accordingly ${Y_2 \triangleq Y^{[S_2,C_2]}}$, such that \[H(\mathbf{A}_1^{[W,S,C]}|\mathbf{Q}_1^{[W,S,C]},\mathbf{X}_W,\mathbf{Y}) = H(\mathbf{A}_1^{[W_2,S_2,C_2]}|\mathbf{Q}_1^{[W_2,S_2,C_2]},\mathbf{X}_W,\mathbf{Y}).\] 

We consider two cases as follows separately: (ii.1) $\mathbf{X}_{W_2}$ is independent of $({\mathbf{Q}^{[W_2,S_2,C_2]}},\mathbf{X}_W,\mathbf{Y},\mathbf{Y}_2)$, and (ii.2) $\mathbf{X}_{W_2}$ depends on $({\mathbf{Q}^{[W_2,S_2,C_2]}},\mathbf{X}_W,\mathbf{Y},\mathbf{Y}_2)$. In the case (ii.1), the proof follows the exact same line as in the proof of case (i), and hence not repeated.

In the case (ii.2), one can readily verify that $X_{W_2}$ must be recoverable from ${Q^{[W_2,S_2,C_2]},X_W,Y,Y_2}$ since $\mathbf{X}_{W_2}$ depends on $({\mathbf{Q}^{[W_2,S_2,C_2]}},\mathbf{X}_W,\mathbf{Y},\mathbf{Y}_2)$. As a result,  
${Y_2=c''_{W_2} X_{W_2}+c''(c_{W_1} X_{W_1}+Z)}$ for some $c''_{W_2},c''\in \mathbb{F}^{\times}_q$. It is also easy to verify that $X_{W_2}$ is not recoverable from $X_{W_1},Y_1,Y_2$, and $\mathbf{X}_{W_2}$ is independent of $(\mathbf{X}_{W_1},\mathbf{Y}_1,\mathbf{Y}_2)$. On the other hand, we have
\begin{align}
D &\geq L+H({\mathbf{A}_1^{[W_1,S_1,C_1]}}|{\mathbf{Q}_1^{[W_1,S_1,C_1]}},\mathbf{X}_{W_1},\mathbf{Y}_1)\label{eq:line15}\\ &= L+H({\mathbf{A}_1^{[W_2,S_2,C_2]}}|{\mathbf{Q}_1^{[W_2,S_2,C_2]}},\mathbf{X}_{W_1},\mathbf{Y}_1).\nonumber
\end{align}
where~\eqref{eq:line15} follows from~\eqref{eq:line29} that holds for $W_1,S_1,C_1$ (and $Y_1$), because $D$ is defined as the maximum total entropy of answers from all servers over all ${W'\in [K]},{S'\in \mathcal{S}},{C'\in \mathcal{C}}$ such that ${W'\in S'}$. Similarly as before, by the server-symmetry assumption 
it can also be shown that
\begin{align}
D &\geq L+\frac{1}{N}H({\mathbf{A}^{[W_2,S_2,C_2]}}|{\mathbf{Q}^{[W_2,S_2,C_2]}},\mathbf{X}_{W_1},\mathbf{Y}_1)\nonumber\\&\geq L+\frac{1}{N}H({\mathbf{A}^{[W_2,S_2,C_2]}}|{\mathbf{Q}^{[W_2,S_2,C_2]}},\mathbf{X}_{W_1},\mathbf{Y}_1,\mathbf{Y}_2).\label{eq:lineseclast}
\end{align} 

Since $\mathbf{X}_{W_2}$ is independent of $(\mathbf{Q}^{[W_2,S_2,C_2]},\mathbf{X}_{W_1},\mathbf{Y}_1,\mathbf{Y}_2)$, and $X_{W_2}$ is recoverable from $A^{[W_2,S_2,C_2]}$, $Q^{[W_2,S_2,C_2]}$, and $Y_2$, a simple application of the chain rule of entropy yields 
\begin{equation}\label{eq:linelast}
H({\mathbf{A}^{[W_2,S_2,C_2]}}|{\mathbf{Q}^{[W_2,S_2,C_2]}},\mathbf{X}_{W_1},\mathbf{Y}_1,\mathbf{Y}_2)\geq L.    
\end{equation} By~\eqref{eq:lineseclast} and~\eqref{eq:linelast}, $D\geq L+L/N$, as was to be shown. 

\subsection{Achievability Proof for Theorem~\ref{thm:PIRCSI-II}}

In this section, we propose a server-symmetric PIR-CSI--II protocol for each $2\leq M\leq K-1$ that achieves a rate equal to ${\mathrm{C}_{W-\text{\it II}}}$ for the corresponding $M$. 

For $3 \leq M \leq K-1$, we assume that each message consists of ${N^2}$ symbols over $\mathbb{F}_q$. 
For $M=2$ and $M=K$, each message can be as short as one $\mathbb{F}_q$-symbol. \vspace{0.25cm}

\textbf{Multi-Server PIR-CSI--II Protocols:}\vspace{0.125cm} 

\textit{{Case of $M=2$:}}
The user randomly selects one of the two indices, say $i$, in $S$ as follows: $i=W$ with probability $1/K$, and $i=S\setminus W$ with probability $(K-1)/K$. Then, the user requests the message $X_i$ from a randomly chosen server. 


\textit{{Case of $3\leq M\leq K-1$:}} The proposed scheme for this case consists of four steps. In the first step, given $W,S,C$ (and $Y^{[S,C]}$), the user utilizes the scheme of~\cite{heidarzadeh2018capacity} for single-server PIR-CSI (under Model~II), which we refer to as \emph{Modified Randomized Partitioning (MRP)}, to construct two sequences $I_1,I_2$ of indices in $[K]$, each of length $M-1$, and two sequences $I'_1,I'_2$ of elements in $\mathbb{F}^{\times}_q$, each of length $M-1$. (For details, see~\cite[Section~V-B]{heidarzadeh2018capacity}.) Next, the user and the servers follow the steps 2-4 of the Multi-Server PIR-CSI--I protocol. 

\vspace{0.05cm}  

\textit{{Case of $M=K$:}} Assume, w.l.o.g., that $W = {1}$. The user randomly chooses an element $c'_1$ from $\mathbb{F}^{\times}_q\setminus \{c_1\}$, where $c_1$ is the coefficient of $X_{1}$ in the side information $Y^{[S,C]}$. Then, the user requests the linear combination $c'_1X_1+c_2X_2+\dots+c_KX_K$ from a randomly chosen server, where $c_i$ is the coefficient of $X_i$ in the side information $Y^{[S,C]}$. 

\textbf{Example 2.} (Case of ${3\leq M\leq \frac{K}{2}+1}$) Assume that there are ${N = 2}$ servers, ${K=10}$ messages from $\mathbb{F}_{3^4}$ (i.e., each message has $4$ symbols over $\mathbb{F}_3$), and ${M=4}$. Suppose that the user demands the message $X_1$ and has a coded side information $X_1+X_2+2X_3+X_4$. Note that, for this example, $W = 1$, ${S= \{1,2,3,4\}}$, and ${C=\{1,1,2,1\}}$. 

First, the user labels $2$ sequences as $I_1,I_2$, each of length $M-1=3$. For creating these sequences, the user selects $w\in \{0,1\}$ and $t\in\{2,3\}$ randomly chosen indices from $W=\{1\}$ and $T = \{5,6,7,8,9,10\}$, respectively, according to a carefully designed probability distribution (ensuring $W$-privacy of the MRP scheme) on all $(w,t)$ such that $w+t= 3$. For this example, the probability distribution is given by
\begin{align*}
p(w,t) \triangleq 
\begin{cases}
0.4, & w=0,t=3\\
0.6, & w=1,t=2\\
\end{cases}  
\end{align*} 
Suppose that the user chooses ${w=1},{t=2}$, and selects the $3$ indices $\{1,6,10\}$. Following the MRP protocol, the user forms the sequence $I_1=S \setminus W=\{2,3,4\}$ and ${I_2=\{1,6,10\}}$. 

The user then labels $2$ sequences as $I'_1,I'_2$, each of length $3$. For this example, suppose that the user creates ${I'_1=\{1,2,1\}, I'_2=\{1,1,1\}}$. Then, the user randomly reorders the elements of $I_i$ and $I'_i$, and constructs
\begin{align*}
\tilde{I}_1 = \{3,2,4\}, &\quad  \tilde{I}'_1 = \{2,1,1\}\\ 
\tilde{I}_2 = \{1,10,6\}, &\quad \tilde{I}'_2 = \{1,1,1\}.
\end{align*}
Next, the user sends a uniform random permutation of $\{I^{*}_1,I^{*}_2\}$, say $\{I^{*}_2,I^{*}_1\}$, to both servers, where $I^{*}_i = (\tilde{I}_i,\tilde{I}'_i)$. The user and the servers form two super-messages as follows: 
\begin{align*}
\hat{X}_1 &=X_1+X_{10}+X_6\\
\hat{X}_2 &=2X_3+X_2+X_4.
\end{align*}

\vspace{-0.125cm}
Finally, the user and the servers run the Sun-Jafar protocol as follows for the two super-messages $\hat{X}_1,\hat{X}_2$ in such a way that the user can privately download $\hat{X}_2$. For each $\hat{X}_i$, let $[\hat{X}_{i,1},\dots,\hat{X}_{i,4}]$ be an independent and uniform random permutation of the $4$ symbols (over $\mathbb{F}_3$) of $\hat{X}_i$. The user requests $3$ symbols from the first server and $3$ symbols from the second server as listed in Table~\ref{table:2}~\cite{sun2017capacity}, where the requested symbols are carefully designed linear combinations of symbols $\{\hat{X}_{i,j}\}_{i\in [2],j\in [4]}$. From the servers' answers, the user first obtains the super-message $\hat{X}_2=2X_3+X_2+X_4$, and then recovers the desired message $X_1$ by subtracting off $\hat{X}_2$ from the side information $X_1+X_2+2X_3+X_4$. 
For this example, the proposed protocol requires to download a total of $6$ symbols (over $\mathbb{F}_3$), achieving the rate of $4/6=2/3$.

\begin{table}[t!]
\caption{The queries/answers of Sun-Jafar protocol for $2$ servers and $2$ messages $\hat{X}_1,\hat{X}_2$, when the user demands $\hat{X}_{2}$\cite{sun2017capacity}.}\vspace{-0.125cm}
    \centering
    \scalebox{1.15}{
\begin{tabular}{ |c|c| } 
 \hline
 Server 1 & Server 2 \\
 \hline
 $\hat{X}_{1,1}$ & $\hat{X}_{1,2}$ \\ [0.5ex]
 $\hat{X}_{2,1}$ & $\hat{X}_{2,2}$ \\ [0.5ex]
 $\hat{X}_{2,3}+\hat{X}_{1,2}$ & $\hat{X}_{2,4}+\hat{X}_{1,1}$ \\ 
 \hline
\end{tabular}}\vspace{-0.35cm}
\label{table:2}
\end{table}

\textbf{Example 3.} (Case of ${\frac{K}{2}\leq M\leq K-1}$) Assume that there are ${N = 2}$ servers, ${K=5}$ messages from $\mathbb{F}_{3^4}$ (i.e., each message has $4$ symbols over $\mathbb{F}_3$), and ${M=4}$. Suppose that the user demands the message $X_1$ and has coded a side information $X_1+X_2+2X_3+X_4$. Note that, for this example, $W = 1$, ${S= \{1,2,3,4\}}$, and ${C=\{1,1,2,1\}}$. 

First, the user labels $2$ sequences as $I_1,I_2$, each of length $M=4$. For creating these sequences, the user selects ${w\in \{0,1\}}$ and $t\in\{2,3\}$ randomly chosen indices from $W=\{1\}$ and $T = \{2,3,4\}$, respectively, according to a carefully designed probability distribution (ensuring $W$-privacy of the MRP scheme) on all $(w,t)$ such that $w+t= 3$. For this example, the probability distribution is given by
\begin{align*}
p(w,t) \triangleq 
\begin{cases}
0.4, & w=0,t=3\\
0.6, & w=1,t=2\\
\end{cases}  
\end{align*} 
Suppose that the user chooses ${w=1},{t=2}$, and selects the $3$ indices $\{1,2,4\}$. Following the MRP protocol, the user forms the sequence $I_1=S=\{1,2,3,4\}$ and ${I_2=\{5,1,2,4\}}$. 

The user then labels $2$ sequences as $I'_1,I'_2$, each of length $4$. For this example, suppose that the user creates ${I'_1=\{2,1,2,1\}, I'_2=\{1,2,1,1\}}$. Then, the user randomly reorders the elements of $I_i$ and $I'_i$, and constructs
\begin{align*}
\tilde{I}_1 = \{1,4,2,3\}, &\quad  \tilde{I}'_1 = \{2,1,1,2\}\\ 
\tilde{I}_2 = \{1,5,2,4\}, &\quad \tilde{I}'_2 = \{2,1,1,1\}.
\end{align*}
Next, the user sends a uniform random permutation of $\{I^{*}_1,I^{*}_2\}$, say $\{I^{*}_2,I^{*}_1\}$, to both servers, where $I^{*}_i = (\tilde{I}_i,\tilde{I}'_i)$. The user and the servers form two super-messages as follows: 
\begin{align*}
\hat{X}_1 &=2X_1+X_5+X_2+X_4\\
\hat{X}_2 &=2X_1+X_4+X_2+2X_3.
\end{align*}

Finally, the user and the servers run the Sun-Jafar protocol as explained in the previous example for the two super-messages $\hat{X}_1,\hat{X}_2$ in such a way that the user can privately download $\hat{X}_2$. The user requests $3$ symbols from the first server and $3$ symbols from the second server as listed in Table~\ref{table:2}~\cite{sun2017capacity}. From the servers' answers, the user first obtains the super-message $\hat{X}_2=2X_1+X_4+X_2+2X_3$, and then recovers the desired message $X_1$ by subtracting off the side information $X_1+X_2+2X_3+X_4$ from $\hat{X}_2$. For this example, the proposed protocol requires to download a total of $6$ symbols (over $\mathbb{F}_3$), achieving the rate of $4/6=2/3$.

\begin{lemma}
The Multi-Server PIR-CSI--II protocols for the cases of ${M=2}$, ${3\leq M\leq K-1}$, and ${M=K}$ are server-symmetric protocols that satisfy the recoverability and the $W$-privacy conditions, and achieve the rates $1$, $N/(N+1)$, and $1$, respectively.
\end{lemma}

\begin{proof}
The proof is similar to the proof of Lemma \ref{lemma 2}, and hence omitted to avoid repetition. 
\end{proof}

\bibliographystyle{IEEEtran}
\bibliography{QGTRefs}

\begin{thebibliography}{10}
\providecommand{\url}[1]{#1}
\csname url@samestyle\endcsname
\providecommand{\newblock}{\relax}
\providecommand{\bibinfo}[2]{#2}
\providecommand{\BIBentrySTDinterwordspacing}{\spaceskip=0pt\relax}
\providecommand{\BIBentryALTinterwordstretchfactor}{4}
\providecommand{\BIBentryALTinterwordspacing}{\spaceskip=\fontdimen2\font plus
\BIBentryALTinterwordstretchfactor\fontdimen3\font minus
  \fontdimen4\font\relax}
\providecommand{\BIBforeignlanguage}[2]{{%
\expandafter\ifx\csname l@#1\endcsname\relax
\typeout{** WARNING: IEEEtran.bst: No hyphenation pattern has been}%
\typeout{** loaded for the language `#1'. Using the pattern for}%
\typeout{** the default language instead.}%
\else
\language=\csname l@#1\endcsname
\fi
#2}}
\providecommand{\BIBdecl}{\relax}
\BIBdecl

\bibitem{sun2017capacity}
H.~Sun and S.~A. Jafar, ``The capacity of private information retrieval,''
  \emph{IEEE Trans.~on Info.~Theory}, vol.~63, no.~7, pp. 4075--4088, 2017.

\bibitem{banawan2018capacity}
K.~Banawan and S.~Ulukus, ``The capacity of private information retrieval from
  coded databases,'' \emph{IEEE Trans.~on Info.~Theory}, vol.~64, no.~3, pp.
  1945--1956, 2018.

\bibitem{kadhe2017private}
S.~Kadhe, B.~Garcia, A.~Heidarzadeh, S.~El~Rouayheb, and A.~Sprintson,
  ``Private information retrieval with side information: The single server
  case,'' in \emph{55th Annual Allerton Conference on Commun., Control, and
  Computing (Allerton)}, 2017, pp. 1099--1106.

\bibitem{heidarzadeh}
A.~Heidarzadeh, B.~Garcia, S.~Kadhe, S.~El~Rouayheb, and A.~Sprintson, ``On the
  capacity of single-server multi-message private information retrieval with
  side information,'' in \emph{Proc. 56th Annual Allerton Conference on
  Commun., Control, and Computing}, Oct 2018, pp. 180--187.

\bibitem{li2018single}
S.~Li and M.~Gastpar, ``Single-server multi-message private information
  retrieval with side information,'' in \emph{56th Annual Allerton Conference
  on Commun., Control, and Computing (Allerton)}, 2018, pp. 173--179.

\bibitem{heidarzadeh2018capacity}
A.~Heidarzadeh, F.~Kazemi, and A.~Sprintson, ``Capacity of single-server
  single-message private information retrieval with coded side information,''
  in \emph{Proc. IEEE Info.~Theory Workshop (ITW'18)}, Nov 2018.

\bibitem{heidarzadeh2019capacity}
\BIBentryALTinterwordspacing
------, ``Capacity of single-server single-message private information
  retrieval with private coded side information,'' Jan 2019. [Online].
  Available: \url{arXiv:1901.09248}
\BIBentrySTDinterwordspacing

\bibitem{tandon2017capacity}
R.~Tandon, ``The capacity of cache aided private information retrieval,'' in
  \emph{55th Annual Allerton Conference on Commun., Control, and Computing},
  Oct 2017, pp. 1078--1082.

\bibitem{wei2018fundamental}
Y.-P. Wei, K.~Banawan, and S.~Ulukus, ``Fundamental limits of cache-aided
  private information retrieval with unknown and uncoded prefetching,''
  \emph{IEEE Trans.~on Info.~Theory}, 2018.

\bibitem{li2018converse}
S.~Li and M.~Gastpar, ``Converse for multi-server single-message pir with side
  information,'' \emph{arXiv preprint arXiv:1809.09861}, 2018.

\bibitem{chen2017capacity}
Z.~Chen, Z.~Wang, and S.~Jafar, ``The capacity of private information retrieval
  with private side information,'' \emph{arXiv preprint arXiv:1709.03022},
  2017.

\bibitem{shariatpanahi2018multi}
S.~P. Shariatpanahi, M.~J. Siavoshani, and M.~A. Maddah-Ali, ``Multi-message
  private information retrieval with private side information,'' in \emph{Proc.
  IEEE Info.~Theory Workshop (ITW)}, Nov 2018.

\bibitem{kazemi2019single}
F.~Kazemi, E.~Karimi, A.~Heidarzadeh, and A.~Sprintson, ``Single-server
  single-message online private information retrieval with side information,''
  \emph{arXiv preprint arXiv:1901.07748}, 2019.

\end{thebibliography}
\end{document}